\documentclass[10pt,conference]{IEEEtran}
\IEEEoverridecommandlockouts
\usepackage{amsmath,amsfonts}
\usepackage{algorithmic}
\usepackage{algorithm}
\usepackage{array}
\usepackage[caption=false,font=footnotesize,labelfont=rm,textfont=rm]{subfig}
\usepackage{textcomp}
\usepackage{stfloats}
\usepackage{url}
\usepackage{verbatim}
\usepackage{graphicx}
\usepackage{cite}
\usepackage{epsfig,latexsym}
\usepackage{float}
\usepackage{indentfirst}
\usepackage{amssymb}
\usepackage{amsmath}
\usepackage{amsthm}
\usepackage{times}
\usepackage{psfrag}
\usepackage{cite}
\usepackage{cases}
\usepackage{epsfig,latexsym}
\usepackage{float}
\usepackage{indentfirst}
\usepackage{amsmath}
\usepackage{amssymb}
\usepackage{times}
\usepackage{subfloat}
\usepackage{psfrag}
\usepackage{cite}
\usepackage{algorithmic,algorithm}
\usepackage{graphicx}
\usepackage{xcolor,multirow,cases,array}
\usepackage{epstopdf}
\usepackage{bm}
\usepackage{color}
\usepackage{booktabs}
\usepackage{multirow}
\usepackage{lipsum,afterpage,refcount}
\usepackage{dsfont}
\newtheorem{theorem}{$\mathbf{Theorem}$}

\newtheorem{assumption}{Assumption}

\def\BibTeX{{\rm B\kern-.05em{\sc i\kern-.025em b}\kern-.08em
    T\kern-.1667em\lower.7ex\hbox{E}\kern-.125emX}}
\begin{document}

\title{A Federated Fine-Tuning Paradigm of Foundation Models in Heterogenous Wireless Networks}
%Towards Online Fine-Tuning for Edge Intelligence: Paradigm Design in O-RAN
% An Online Federated Fine-Tuning Paradigm for Edge Intelligence in O-RAN Architecture
\author{Jingyi Wang\IEEEauthorrefmark{1}, Zhongyuan Zhao\IEEEauthorrefmark{2}, Qingtian Wang\IEEEauthorrefmark{1}, Zexu Li\IEEEauthorrefmark{1}, Yue Wang\IEEEauthorrefmark{1}, Tony Q. S. Quek\IEEEauthorrefmark{3}\\
		\IEEEauthorrefmark{1}Wireless AI System Research Team, China Telecom Research Institute, Beijing, 102209, China\\
        \IEEEauthorrefmark{2}The State Key Laboratory of Networking and Switching Technology\\
		Beijing University of Posts and Telecommunications, Beijing, 100876, China\\
        \IEEEauthorrefmark{3}Singapore University of Technology and Design, 487372, Singapore\\
		E-mail: wangjy74@chinatelecom.cn, zyzhao@bupt.edu.cn, \{wangqt08, lizx28, yue.wang\}@chinatelecom.cn, \\ tonyquek@sutd.edu.sg\vspace*{-3mm}}

%\author{Jingyi Wang, Qingtian Wang, Zexu Li, Bei Yang, Yue Wang \\
  %China Telecom Research Institute, Beijing, 102209, China\\
  %E-mail: \{wangjy74, wangqt08, lizx28, yangbei1, yue.wang\}@chinatelecom.cn}

\maketitle
\begin{abstract}
Edge intelligence has emerged as a promising strategy to deliver low-latency and ubiquitous services for mobile devices. Recent advances in fine-tuning mechanisms of foundation models have enabled edge intelligence by integrating low-rank adaptation (LoRA) with federated learning. However, in wireless networks, the device heterogeneity and resource constraints on edge devices pose great threats to the performance of federated fine-tuning. To tackle these issues, we propose to optimize federated fine-tuning in heterogenous wireless networks via online learning. First, the framework of switching-based federated fine-tuning in wireless networks is provided. The edge devices switches to LoRA modules dynamically for federated fine-tuning with base station to jointly mitigate the impact of device heterogeneity and transmission unreliability. Second, a tractable upper bound on the inference risk gap is derived based on theoretical analysis. To improve the generalization capability, we formulate a non-convex mixed-integer programming problem with long-term constraints, and decouple it into model switching, transmit power control, and bandwidth allocation subproblems. An online optimization algorithm is developed to solve the problems with polynomial computational complexity. Finally, the simulation results on the SST-2 and QNLI data sets demonstrate the performance gains in test accuracy and energy efficiency.
%, xx%, with xx benchmarks, respectively.
% under the O-RAN architecture, which enables data-driven network intelligent operations by introducing RAN Intelligent Controller (RIC)
\end{abstract}

\IEEEpeerreviewmaketitle

\section{Introduction}
%Artificial intelligence (AI) has been considered as an attracting feature of the future sixth generation (6G) systems \cite{6G}. To integrate the AI algorithms into network management and enable multi-vendor interoperability of 6G, Open RAN (O-RAN) has been proposed \cite{ORAN2}, which introduces RAN Intelligent Controller (RIC) components to perform non real-time (Non-RT) and near real-time (Near-RT) network functions.
%The Non Real-Time (Non-RT) RIC implements the function of model training, data analytics, and policy enforcement, while the Near Real-Time (Near-RT) RIC provides the service of resource allocation, network management and model deployment by hosting the third-party applications, i.e., xAPPs, to perform near real-time network functions.
Benefiting from the potentiality of massive data and enhancement of computing capability on edge devices, network intelligentization is extending from the remote cloud server to the network edge \cite{edge}. Edge intelligence, which is envisioned as a key enabler to the sixth generation (6G) networks, has facilitated real-time processing for applications such as autonomous vehicles, embodied intelligence and augmented reality. %Empowered by the extensive research
Recent advances in fine-tuning on foundation model have provided an efficient mechanism to realize edge intelligence \cite{FT1}, in which pre-trained foundation models adapt to diverse downstream tasks and generate personalized solutions for edge devices.

As the computing resources are decentralized on the edge devices, federated fine-tuning (FedFT) is proposed to improve fine-tuning efficiency and protect user privacy \cite{FedFT}. The edge server periodically collects the fine-tuned foundation models from edge devices and executes global aggregation until convergence \cite{few-shot}. However, full-parameter FedFT may lead to excessive communication and computation overheads for resource-constrained edge devices. To tackle this issue, the low-rank adaptation (LoRA) technique \cite{lora} is proposed as a parameter-efficient fine-tuning strategy, which only updates the LoRA matrix components while freezing the foundation models to reduce the amount of trainable parameters.
%As a promising solution to mitigate the additional inference latency, the low-rank adaptation (LoRA) technique \cite{lora} is proposed with enhanced model quality and computation efficacy.

Nonetheless, deploying FedFT for edge devices should address the challenge of device heterogeneity with non-independent
and identical distributed (non-IID) data characteristics, and diverse computation and communication resource constraints \cite{hetlora}. The recent works mainly focused on adjusting LoRA rank and enhancing aggregation strategy. In \cite{hetlora}, the local fine-tuning matrices are aggregated in the same dimension by zero-padding and truncated towards a personalized LoRA rank.
%In \cite{slora}, a data-driven strategy is employed by first sparse fine-tuning for several epochs as initialization and then deploying LoRA to improve the robustness of federated parameter efficient fine-tuning.
In \cite{FlexLora}, a dynamic LoRA rank is employed for device fine-tuning by adjusting the parameter allocation strategy for weight redesign.
In \cite{HierFlexLora}, devices are divided into near-IID groups for FedFT, and the intra-group aggregation frequency and fine-tuning depth are jointly optimized.

Although the recent approaches are effective in mitigating device heterogeneity, few works are engaged in the context of wireless networks, and the quantitative analysis of FedFT on device heterogeneity and resource constraints has not been studied in the literature. Meanwhile, the LoRA configurations are assumed to be uniform, which may degrade the performance under time-varying wireless circumstances. The dynamic FedFT schemes, including model switching-based FedFT, show great potentials for practical application. Motivated by these challenges, we propose a model switching-based federated fine-tuning paradigm in heterogenous wireless networks. Our main contributions are summarized as follows:
%In this section, we derive an upper bound on the inference accuracy gap between the ensemble model and the optimal model over the test data distribution of our proposed framework.
%However, due to distribution divergence, the accuracy performance of federated learning model trained in a specific cell can not be guaranteed on new wireless settings.
\begin{itemize}
    \item \textit{First}, the framework of federated fine-tuning with model switching is proposed. The foundation models are customized with diverse LoRA modules, and the edge devices implement switching-based federated fine-tuning to mitigate the impact of device heterogeneity.
        % , models from Near-RT RIC and  under computation and communication constraints.
        %The pre-trained GenAI models generated on the Non-RT RIC are customized on the Near-RT RIC with LoRA modules and subscribed by UEs to implement fine-tuning in a federated manner under computation and communication constraints.
         %unreliable transmission circumstances. % constrained computation resources and based on the online feedbacks from UEs
    \item \textit{Second}, a tractable upper bound on the generalization performance is derived with regard to device heterogeneity, and an online optimization algorithm is designed on model switching and wireless resource management.
    \item \textit{Finally}, the simulation results prove the validity of theoretical analysis and demonstrate the performance gains.
\end{itemize}

\section{System Model of Federated Fine-tuning}
% In the heterogenous networks, $N$ AI/ML models, denoted by $\mathcal{M} = \{\mathcal{M}_1, \cdots, \mathcal{M}_N \}$, are employed to meet the various QoS requirements of UEs. We assume that the $N$ models are stored on Non-RT RIC, serving as a cloud market. Each UE can purchase multiple models from the cloud market, and deploy the models for local inference. The Near-RT RIC are responsible for setting the prices of available models.
In Fig. \ref{fig: framework}, we focus on the scenario of federated fine-tuning based on LoRA in the wireless edge network. The gNodeB (gNB) is located at the center of its Voronoi cell with $K$ associated user equipments (UEs) $\mathcal{U} = \{U_1, \cdots, U_K\}$, and equipped with a computing server to enable data processing and interaction with devices. Each UE $U_k \in \mathcal{U}$ is equipped with a computing unit for model fine-tuning based on its local data set $\mathcal{D}_k$, which follows a non-IID probability density function. The gNB customizes $N$ LoRA modules, denoted by $\Delta \mathcal{W} = \{\Delta \mathbf{w}_{1}, \cdots, \Delta \mathbf{w}_{N}\}$, based on foundation model $\mathbf{w}_0$.

\begin{figure}[htb]\centering
\includegraphics[width = 0.47\textwidth]{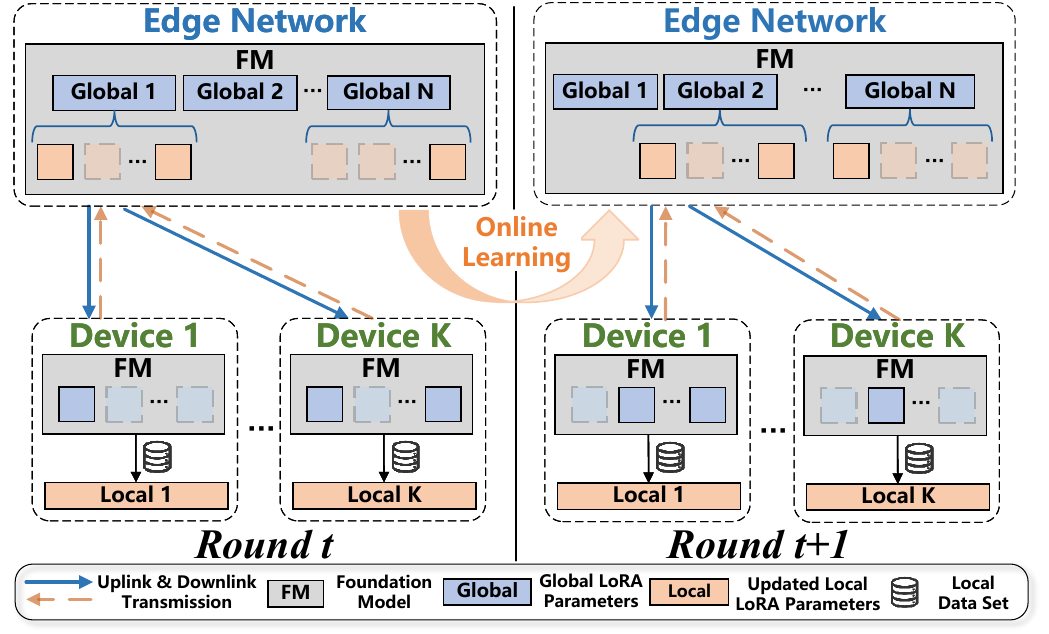}
\caption{The framework of federated fine-tuning scheme with model switching.}
\label{fig: framework}\vspace*{-3mm}
\end{figure}

In order to fine-tune foundation models in a dynamic manner, the model switching mechanism is deployed. In each round, the UEs individually subscribe LoRA modules from gNB to implement local fine-tuning by employing gradient-descent methods, and then transmit the feedback gradients to the gNB for federated learning.
%UEs subscribe LoRA modules to implement online fine-tuning based on the few-shot data set, and the subscription results vary with $t$.
% The method of model switching will be illustrated in Section IV.
% , which is assumed to be non-independent and identically distributed (non-IID).
% \Delta \mathbf{w}_{n,k}^{t+e} = {\Delta \mathbf{w}}_{n}^t - \eta \nabla F_k(\Delta \mathbf{w}_{n}^t;\mathcal{D}_k^{t:t+e}) = \Delta \mathbf{w}_{n}^t -  \frac{\eta}{|\mathcal{D}_k^t|} \sum_{d=1}^{|\mathcal{D}_k^t|} l(\mathbf{w}_{n}^t;\mathbf{x}^d,\mathbf{y}^d),
%\begin{equation}\label{eqn:loss}
% F_{k}(\Delta \mathbf{w}_{n}^t; \mathcal{D}_{k}^{t}) =  \frac{1}{D_{k}^{t}} \sum_{d=1}^{D_{k}^{t}} l(\Delta \mathbf{w}_{n}^{t};\mathbf{x}^d,\mathbf{y}^d),
%\end{equation}
%It should be noted that in each inference time, the local model is updated accordingly, which will be used for inference on subsequent data samples.
%Since UEs are not equipped with sufficient computing capacity for local model update, an asynchronous federated fine-tuning strategy is employed, in which the global model is fine-tuned based on the gradients from multiple UEs in a federated manner.
%\subsubsection{Online federated fine-tuning}
%As model fine-tuning may suffer from performance degradation only based on the few-shot data samples per UE \cite{few-shot}, we propose an online LoRA-based federated fine-tuning framework
%to integrate the personalized knowledge of downstream target tasks while preserving data privacy and improve online generalization capability.
In specific, the model parameters of the $n$-th LoRA module in the $t$-th round, denoted by $\Delta \mathbf{w}_{n}^t$, is updated at gNB based on the aggregation of feedback gradients from UEs in subscription, i.e.,
% Thus, the transmission overheads for UEs can be significantly lowered.
\begin{equation}\label{eqn:global_model}
\Delta \mathbf{w}_{n}^{t+1} = \Delta \mathbf{w}_{n}^t - \eta  \frac{\sum_{k=1}^K D_k^t \beta_{n,k}^t \gamma_{n,k}^t \nabla  F_{k}(\Delta \mathbf{w}_{n}^t; \mathcal{D}_{k}^{t})}{\sum_{k=1}^K D_k^t \beta_{n,k}^t \gamma_{n,k}^t},
\end{equation}
where $\nabla  F_{k}(\Delta \mathbf{w}_{n}^t; \mathcal{D}_{k}^{t}) = \frac{1}{D_{k}^{t}} \sum_{d=1}^{D_{k}^{t}} \nabla l(\Delta \mathbf{w}_{n}^t;\mathbf{x}^d,\mathbf{y}^d)$ denotes the gradient of empirical risk of $U_k$ for local model update on the $n$-th LoRA module in the $t$-th round, $D_k^{t}$ denotes the volume of fine-tuning data set randomly sampled from $\mathcal{D}_k$ in the $t$-th round, $l(\Delta \mathbf{w};\mathbf{x},\mathbf{y})$ is the risk function of $\Delta \mathbf{w}$ on a specific data sample $(\mathbf{x},\mathbf{y})$, $\beta_{n,k}^t$ is a binary indicator, in which $\beta_{n,k}^t =1$ when $U_k$ has subscribed the $n$-th LoRA module in the $t$-th round and $\beta_{n,k}^t = 0$ otherwise, and $\eta$ is the step size.

To characterize the uplink channel quality for transmitting $\nabla F_k(\mathbf{w}_{n}^t; \mathcal{D}_{k}^{t})$ to gNB, $\gamma_{n,k}^t$ is defined as
\begin{equation}\label{eqn:suc_recovery}
\gamma_{n,k}^{t}  = \mathds{1} \{\mathbf{SINR} =\frac{|\mathbf{h}_k^t|^2 P_k^t d_k^{-\alpha}}{\sum_{U_{l} \in \tilde{\Xi}_{o,k}} |\mathbf{h}_l^t|^2 P_l^t d_l^{-\alpha} + \sigma^2} \geq \theta  \},
\end{equation}
where $\mathbf{h}_k^t$ captures the small-scale flat Rayleigh channel fading, $P_k^t$ is the transmit power of $U_k$ in the $t$-th round, $d_k^{-\alpha}$ denotes the path loss with exponent $\alpha$, $\tilde{\Xi}_{o,k} $ is the location set of out-of-cell UEs that interfere with $U_k$, $\theta$ is the threshold of signal-to-interference-plus-noise ratio (SINR) at receiver, $\sigma^2= W \delta_{n,k}^t N_0$ is the power of noise, $W$ denotes the total bandwidth resource, $\delta_{n,k}^t$ denotes the bandwidth allocation coefficient of $U_k$ for transmission of $\nabla F_{k}(\Delta \mathbf{w}_{n}^t; \mathcal{D}_{k}^{t})$, $0 \leq \delta_{n,k}^t \leq 1$, and $N_0$ is the noise power density.

As $\nabla F_{k}(\Delta \mathbf{w}_{n}^t; \mathcal{D}_{k}^{t})$ is transmitted via orthogonal channels, the intra-cell interference can be eliminated. The gNB first broadcasts back the intact $\mathbf{w}_0$ for UEs updating LoRA modules in each round, and iteratively implements federated learning until $\Delta \mathbf{w}_{n}^t$ converge. In the inference phase, by summing the LoRA module parameters with the foundation model parameters \cite{lora}, the inference risk of $U_k$ on the test data set $\mathcal{D}_{k, \mathrm{te}}^{t}$ via model ensemble can be modelled as\footnote{For simplicity of analysis, we employ a uniform weight for model ensemble. The sophisticated weight design for harmonizing non-IID data distribution divergence can follow the work in \cite{noniid}.} %The sophisticated weight design for harmonizing non-IID data distribution divergence can be found in \cite{noniid}.
\begin{equation}\label{eqn:inference}
\bar{F}_k(\mathbf{w}_t; \mathcal{D}_{k, \mathrm{te}}^{t}) = \frac{1}{N} \sum_{n=1}^N F_k(\Delta \mathbf{w}_{n}^{t+1}+\mathbf{w}_0; \mathcal{D}_{k, \mathrm{te}}^{t}).
\end{equation}

%The gradient is transmitted by each UE via orthogonal channel, and thus the interference can be eliminated.

% (***)In each round, UEs will transmit $\Delta \mathbf{w}_{n}$ instead of $\mathbf{w}_{0}$ to the server, and thus the resource overheads can be significantly lowered. % additive white Gaussian
The energy consumption of $U_k$ in the $t$-th round of federated learning consists of both computation and communication costs, which can be formulated as% from the perspective of gradient descent-based computation and uplink transmission
\begin{equation}\label{eqn:energy_UE}
C_{n,k}(t) =  \tau D_{k}^t G_{n,k}^{\mathbf{w}} f_{k}^2 \varrho_{k} M_{k} + P_k^t \frac{G_{n,k}^{\Delta \mathbf{w}}}{W \delta_{n,k}^t \log_2(1+ \mathbf{SINR} )} ,
\end{equation}
where $\tau$ is the number of local iterations, $f_k$ is the CPU/GPU frequency of $U_k$, $G_{n,k}^{\mathbf{w}}$ and $G_{n,k}^{\Delta \mathbf{w}}$ denote the data size of $\mathbf{w}_{0}$ and $\nabla F_{k}(\Delta \mathbf{w}_{n}^t; \mathcal{D}_{k}^{t})$ measured in bits, respectively, $M_k$ is the number of CPU/GPU cycles required to process one bit at $U_k$, and $\varrho_k$ is the energy consumption coefficient of the chip.

\section{Generalization Performance Analysis}
To analyze the generalization capability of the proposed federated fine-tuning paradigm, we study the expected risk gap between the inference risk in \eqref{eqn:inference} and the minimum risk over the downstream tasks in the $t$-th round, i.e.,
% \footnote{The expectation is performed with regard to channel variations, i.e., $\mathbf{h}_k^t$.}     %In the $t$-th frame, the data volume for fine-tuning has been determined.
\begin{equation}\label{eqn:target}
\Phi_t = \sum_{k=1}^K \rho_k^t  \mathbb{E}_{\mathbf{h}} \{\bar{F}_k(\mathbf{w}_t; \mathcal{D}_{k, \mathrm{te}}^{t}) - F_k(\mathbf{w}_{k,t}^{*}; \mathcal{D}_{k, \mathrm{te}}^{t}) \},
\end{equation}
where $\rho_k^t = \frac{D_k^t}{\sum_{k=1}^K D_k^t}$ denotes the weight of data volume in the $t$-th round, and $\mathbf{w}_{k,t}^{*}$ denote the optimal model parameters for inference on $\mathcal{D}_{k, \mathrm{te}}^{t}$, i.e., $\mathbf{w}_{k,t}^{*} = \arg\min_{\mathbf{w}} F_k(\mathbf{w} ; \mathcal{D}_{k, \mathrm{te}}^{t})$. For brevity of notation, $F_k(\mathbf{w}; \mathcal{D}_{k, \mathrm{te}}^{t})$ is rewritten as $F_k(\mathbf{w})$. The following assumptions on $l(\mathbf{w};\mathbf{x},\mathbf{y})$ are made \cite{noniid, assum}.
%Based on the definition of Total Variation (TV) distance [XX], the TV distance can be defined as follows.
%\begin{Definition}\label{eqn:TV_def}
%For distributions $p_1$ and $p_2$ over countable set $\mathcal{Z}$, their TV distance is given as $\| p_1-p_2 %\|_{TV}=\frac{1}{2}\sum_{(\mathbf{x},\mathbf{y}) \in \mathcal{Z}} |p_1(\mathbf{x},\mathbf{y})-p_2(\mathbf{x},\mathbf{y})|$.
%\end{Definition}
\begin{assumption}\label{assumpt1}
$l(\mathbf{w};\mathbf{x},\mathbf{y})$ is $\epsilon$-Lipchitz continuous and $\xi$-strongly convex with respect to $\mathbf{w}$.
%, i.e., $\|\nabla l(\mathbf{w}_1;\mathbf{x},\mathbf{y}) - \nabla l(\mathbf{w}_2;\mathbf{x},\mathbf{y}) \| \leq L \| \mathbf{w}_1 - \mathbf{w}_2 \|$.
\end{assumption}
\begin{assumption}\label{assumpt2}
$l(\mathbf{w};\mathbf{x},\mathbf{y})$ is twice-continuously differentiable, i.e., $\xi \boldsymbol{I} \preceq \nabla^2 l(\mathbf{w};\mathbf{x},\mathbf{y}) \preceq \epsilon \boldsymbol{I}$.
%w.r.t. $\mathbf{w}$, i.e., $l(\mathbf{w}_1;\mathbf{x},\mathbf{y}) \geq l(\mathbf{w}_2;\mathbf{x},\mathbf{y}) + (\mathbf{w}_1 - \mathbf{w}_2)^\top \nabla l(\mathbf{w}_2;\mathbf{x},\mathbf{y}) + \frac{\xi}{2}\| \mathbf{w}_1 - \mathbf{w}_2\|^2 $, and $l(\mathbf{w};\mathbf{x},\mathbf{y})$ is twice-continuously differentiable, i.e., $\xi \boldsymbol{I} \preceq \nabla^2 l(\mathbf{w};\mathbf{x},\mathbf{y}) \preceq L \boldsymbol{I}$.
\end{assumption}
\begin{assumption}\label{assumpt3}
There exists $\zeta_{1}$, $\zeta_{2} \geq 0$ that satisfies $\| \nabla l(\mathbf{w};\mathbf{x},\mathbf{y}) \|^2 \leq \zeta_{1} + \zeta_{2}\| \nabla F(\mathbf{w})\|^2$.
\end{assumption}

\begin{figure*}[ht]
\begin{equation}\label{eqn:theorem}\small \vspace*{-3mm}
\begin{aligned}
\Phi_t & \leq \underbrace{\frac{1}{N} \sum_{k=1}^K \rho_k^t \sum_{n=1}^N (1- A_n^{t})[F_k( \Delta \mathbf{w}_{n}^{t})-F_k(\mathbf{w}_{k,t}^{*})]}_{\mathrm{Convergence~of~local~fine-tuning}} + \underbrace{\sum_{q=0}^{t} \prod_{p=q+1}^{t} \sum_{n=1}^{N} ( 1- A_n^p)B_n^{q}}_{\mathrm{Long-term~impact~of~transmission~unreliability}}  \\
& + \underbrace{\frac{\epsilon}{N} \sum_{k=1}^K \rho_k^t \sum_{n=1}^N A_n^{t} \sqrt{\frac{4 \zeta_1}{\xi^2 D_{k, \mathrm{te}}^{t}}} }_{\mathrm{Data~heterogenity~of~UE-specific~tasks}} +\underbrace{\frac{\epsilon+2}{4}  \sum_{k=1}^K \rho_k^t \|\mathbf{w}_0\|^2}_{\mathrm{Intrinsic~impact~of~foundation~model}},
\end{aligned}\vspace*{-5mm}
\end{equation}
\end{figure*}
%\addtolength{\topmargin}{-0.02in}

Theorem \ref{theorem} gives an upper bound of $\Phi_t$ as follows.
\begin{theorem}\label{theorem}
Given the step size as $\eta = \frac{1}{\epsilon}$, there exists an upper bound on the expected risk gap $\Phi_t$ as \eqref{eqn:theorem}, where $\lambda_{n,k} = \mathbb{E}_{\mathbf{h}} \{ \gamma_{n,k}^t \}$ denotes the successful transmission probability given in \eqref{eqn:succ_prob}, and the convergence rate is derived as $A_n^t =\frac{2\xi}{\epsilon}[1 - 8\zeta_2 \sum_{k=1}^K \rho_k^t (1-\lambda_{n,k})(1- \beta_{n,k}^t)-4\zeta_2 K - \frac{\epsilon}{2} ]$, $B_n^t =\frac{2 \zeta_1}{\epsilon}[K+2 \sum_{k=1}^K \rho_k^t (1-\lambda_{n,k})(1- \beta_{n,k}^t)]$. %  $\mathbf{w}^{*}$ denote the global optimal parameters, i.e., $\mathbf{w}^{*} = \arg\min_{\mathbf{w}} F(\mathbf{w})$,
\end{theorem}
\begin{proof}
Please refer to Appendix A.
\end{proof}
Theorem \ref{theorem} illustrates that there exists a tractable upper bound on the generalization risk gap between the proposed federated fine-tuning scheme and the optimal scheme over downstream tasks, which is correlated with the impact of model switching in terms of $\beta_{n,k}^t$, transmission unreliability in terms of $\lambda_{n,k}$, data heterogeneity of UE-specific tasks, and the impact of foundation models. Moreover, the convergence rate of generalization risk gap is ensured, which depends on channel quality and model switching algorithms.

%Therefore, to minimize $\Phi_t$ and reduce the carbon footprints, we will design a joint optimization algorithm on power control, model switching and model integration in the next section.

\section{Online Optimization on Model Switching, Power Control and Bandwidth Allocation}
In order to improve the generalization capability performance, and reduce the energy consumption of the federated fine-tuning system, the objective function is formulated by optimizing on model switching $\bm{\beta}^t=\{\bm{\beta}_1^t, \cdots, \bm{\beta}_K^t\}$, $\bm{\beta}_k^t=\{\beta_{1,k}^t, \cdots, \beta_{N,k}^t\}$, transmit power control $\mathbf{P}^t = \{ P_1^t, \cdots, P_K^t\}$, and resource allocation $\bm{\delta}^t=\{\bm{\delta}_{1}^t, \cdots, \bm{\delta}_{K}^t\}$, $\bm{\delta}_k^t=\{\delta_{1,k}^t, \cdots, \delta_{N,k}^t\}$ as follows % \quad & C_k^{\mathrm{alw }}  \sum_{n=1}^N \mu_{n,k}^t \beta_{n,k}^t = 1, \forall k, \label{eq:sub6}
\begin{subequations}\label{eqn:opt}
\begin{align}
 \min_{\mathbf{P}^t, \bm{\beta}^t, \bm{\delta}^t} \,\,&  \Phi_t + \mu \sum_t \sum_{n=1}^N \sum_{k=1}^K \beta_{n,k}^{t} C_{n,k}(t)  \\
%\quad & \eta_{i,t} \in \{0,1\}, \forall i, \forall t, \\
s.t. & \quad  \beta_{n,k}^t \in \{0,1 \},\forall n, \forall t, \label{eq:sub2} \\
&  \quad  \sum_{n=1}^N \beta_{n,k}^{t} \leq s^t, \forall k, \forall t, \label{eq:sub3} \\
 & \quad \frac{1}{T} \sum_{t} \beta_{n,k}^{t} \geq v_n, \forall n, \label{eq:sub4} \\
&\quad   0 < P_k^t \leq P_{\mathrm{max}},\forall k, \forall t, \label{eq:sub1}\\
% \quad &  E \geq E_{\mathrm{min}},\forall k,\forall t, \label{eq:sub5} \\
% (\mathrm{Bandwidth~constraint})& \quad  0 \leq \delta_{n,k}^t \leq 1, \forall k,\forall t, \label{eq:sub5}\\
& \quad  \sum_{n=1}^N \sum_{k=1}^K \delta_{n,k}^t \leq 1, \delta_{n,k}^t \in [0,1], \label{eq:sub7}
%(\mathrm{Frame~length~constraint})& \quad E^t \geq \max_k \max_n \frac{G_{n,k}}{W \delta_{n,k}^t \log_2(1+ \mathbf{SINR}_{n,k}^t )}, \forall t, \label{eq:sub6}
\end{align}
\end{subequations}
where $\mu$ is the energy efficiency coefficient to balance generalization capability and resource overhead, $s^t$ denotes the maximum number of LoRA modules in subscription, and $v_n$ is the participation rate constraint. Since the optimization problem is a non-convex, mixed-integer programming problem with long-term constraints, we decouple \eqref{eqn:opt} into three subproblems, and develop an online algorithm to solve them iteratively.
%First, we design a distributed optimization method on transmit power control to improve the successful transmission probability for each user while ensuring energy efficiency. Second,
%incentive long-term

% , $P_{\mathrm{max}}$ denotes the maximum power for UEs, $s^t$ denotes the maximum number of adapters that can be subscribed in the $t$-th frame, and $v_n$ denotes the minimum participation rate of the $n$-th adapter

\subsection{Model Switching from An Energy-Efficient Perspective}
Recalling the impact of $\beta_{n,k}^t$ in \eqref{eqn:theorem}, the dynamic model switching method can benefit UEs with enhanced generalization capability and energy efficiency. Therefore, in order to minimize the long-term cost in \eqref{eqn:opt}, the subproblem of model switching for an individual UE $U_k$ is organized as
\begin{subequations}\label{eqn:opt_2}
\begin{align}
& \min_{\bm{\beta}_k^t} \,\,  \sum_{q=0}^{t} \sum_{n=1}^{N} \prod_{p=q+1}^{t} ( 1- A_{n,k}^p) B_{n,k}^q + \mu \sum_{q=0}^{t} \sum_{n=1}^{N} \beta_{n,k}^{t} C_{n,k}(t) \label{eq:sub9a}\\
& s.t. \,\eqref{eq:sub2}, \,\eqref{eq:sub3} \,\, \mathrm{and} \,\, \eqref{eq:sub4},
%\quad & \eta_{i,t} \in \{0,1\}, \forall i, \forall t, \\
\end{align}
\end{subequations}
where $A_{n,k}^t =\frac{2\xi}{\epsilon}[1 - 8\zeta_2 (1-\lambda_{n,k})(1- \beta_{n,k}^t)-4\zeta_2 K - \frac{\epsilon}{2} ]$, $B_{n,k}^t =\frac{2 \zeta_1}{\epsilon}[K+2 (1-\lambda_{n,k})(1- \beta_{n,k}^t)]$.

The integer linear programming problem in \eqref{eqn:opt_2} can be transformed into a linear programming problem by relaxing $\beta_{n,k}^{t}$ to the real domain, i.e., $\beta_{n,k}^{t} \in [0,1], \forall n, \forall t $.
%Based on the derivation of $\Phi_t$ in \eqref{eqn:theorem},
The cost function to be minimized in the $t$-th round is expressed as
\begin{equation}\label{eqn:X}
Q_k^t = \sum_{n=1}^{N} [( 1- A_{n,k}^t) B_{n,k}^{t-1} + B_{n,k}^{t}]+ \mu \sum_{n=1}^{N} \beta_{n,k}^{t} C_{n,k}(t).
\end{equation}

Denote by $y^t(\bm{\beta}_k^{t}) = \sum_{n=1}^N \beta_{n,k}^{t} - s^t$, and $z_n^t = v_n - \beta_{n,k}^{t}$, which satisfies $\mathbf{z}^t(\bm{\beta}_k^{t}) = [z_1^t, \cdots, z_N^t]\preceq 0 $. Thus, the problem in \eqref{eqn:opt_2} can be transformed into its min-max version as
\begin{subequations}\label{eqn:opt_3}
\begin{align}
\min_{\bm{\beta}_k^t} \max_{\bm{\lambda}^t} \,\, & \sum_t \mathcal{F}^t (\bm{\beta}_k^{t}, \bm{\lambda}^t)= \sum_t Q_k^t(\bm{\beta}_k^{t}) + \sum_t \bm{\lambda}^t \mathbf{z}^t(\bm{\beta}_k^{t}) \\
s.t. \quad & y^t(\bm{\beta}_k^{t}) \leq 0, \beta_{n,k}^{t} \in [0,1], \forall n, \forall t,
\end{align}
\end{subequations}
where $\bm{\lambda}^t$ is the Lagrange multiplier, $Q_k^t(\bm{\beta}_k^{t})= \frac{16 \xi \zeta_2 }{\epsilon} \sum_{n=1}^{N} \{ B_{n,k}^{t-1} (1-\lambda_{n,k})(1- \beta_{n,k}^t) + \frac{4 \zeta_1}{\epsilon} (1-\lambda_{n,k})(1- \beta_{n,k}^t)\}+ \mu \sum_{n=1}^{N} \beta_{n,k}^{t} C_{n,k}(t)$, in which $ B_{n,k}^{t-1}$ is obtained in the $(t-1)$-th round.
%The long-term constraint \eqref{eq:sub4} has been absorbed into the objective.
To solve \eqref{eqn:opt_3}, we can alternately update $\bm{\lambda}^{t+1}$ and $\bm{\beta}_k^{t+1}$ on the fly based on the knowledge before the $(t+1)$-th round. The dual variable $\bm{\lambda}^{t+1}$ is updated via standard dual ascent step as
\begin{equation}\label{eqn:opt_4}
\bm{\lambda}^{t+1} = [\bm{\lambda}^{t} + \varsigma \nabla_{\bm{\lambda}} \mathcal{F}^t (\bm{\hat{\beta}}_k^{t}, \bm{\lambda}^t) ]^{+},
\end{equation}
where $\nabla_{\bm{\lambda}} \mathcal{F}^t (\bm{\hat{\beta}}_k^{t}, \bm{\lambda}^t)$ is the gradient of $\mathcal{F}^t (\bm{\hat{\beta}}_k^{t}, \bm{\lambda})$ given $\bm{\lambda}=\bm{\lambda}^t$, $\varsigma$ is step size, and $\bm{\hat{\beta}}_k^{t}$ is the solution obtained at $t$. $\bm{\beta}_k^{t+1}$ is updated via a decent step by approximating $\mathcal{F}^t (\bm{\beta}, \bm{\lambda}^{t+1})$ as
\begin{equation}\label{eqn:opt_4}
\min_{\bm{\beta}} \nabla Q_k^t(\bm{\hat{\beta}}_k^{t})(\bm{\beta} - \bm{\hat{\beta}}_k^{t}) + \bm{\lambda}^{t+1} \mathbf{z}^t(\bm{\beta}), s.t.~ y^t(\bm{\beta}) \leq 0.
\end{equation}

The problem in \eqref{eqn:opt_4} can be solved by employing the interior point method to approach $\kappa$-accurate optimal solution with polynomial time complexity $\mathcal{O}(N^2 \log(1/\kappa))$. The fractional solution $\bm{\hat{\beta}}_k^{t}$ from \eqref{eqn:opt_4} can be further converted into integers through randomized rounding algorithms.\vspace*{-2mm}

\subsection{Wireless Resource Management for Reliable Transmission}
%Recalling \eqref{eqn:X}, $Q_k^t$ is correlated with the tradeoff between transmission reliability and energy consumption. In order to balance the tradeoff and manage wireless resources flexibly, we propose an adaptive transmit power control and bandwidth allocation scheme in a distributed manner.

Due to the sophisticated characteristics of SINR per UE, we first relax $\hat{C}_{n,k}(t)= \tau D_{k}^t G_{n,k}^{\mathbf{w}} f_{k}^2 \varrho_{k} M_{k} + P_k^t E^t$ in \eqref{eqn:energy_UE},
%As the gradients from multiple UEs are transmitted via orthogonal channels in the OFDM system,
where $E^t$ denotes the minimum transmission delay for the specific LoRA module of UE, i.e.,
\begin{equation}\label{eqn:trans_time}
E^t = \min \max_{n,k} \frac{G_{n,k}^{\Delta \mathbf{w}}}{W \delta_{n,k}^t \log_2(1+ \mathbf{SINR} )}.
\end{equation}

Thus, $Q_k^t$ in \eqref{eqn:X} can be reformulated as \eqref{eqn:main_obj}.
\begin{figure*}[ht]
\begin{equation}\label{eqn:main_obj} \small \vspace*{-3mm}
\hat{Q}_k^t = \frac{16\xi \zeta_2 }{\epsilon} \sum_{n=1}^{N} B_{n,k}^{t-1} (1-\lambda_{n,k})(1- \beta_{n,k}^t) + \frac{4\zeta_1}{\epsilon}\sum_{n=1}^{N} (1-\lambda_{n,k})(1- \beta_{n,k}^t) + \mu \sum_{n=1}^N \beta_{n,k}^{t} P_k^t E^t
\end{equation}\normalsize\vspace*{-4mm}
\end{figure*}
$\lambda_{n,k}$ takes the following form by conditioning on Rayleigh distribution as
\begin{equation}\label{eqn:succ_prob}
\lambda_{n,k} = \mathrm{exp}[\frac{-W \delta_{n,k}^t N_0 \theta d_i^2}{P_{k}^t (2 \phi)^{\frac{\alpha}{2}-1}} - \phi \pi d_i^2 \theta^{\frac{2}{\alpha}} \int_0^{+\infty}\frac{1-e^{-\frac{12}{5\pi}\theta^{\frac{2}{\alpha}}x}}{1+x^{\frac{\alpha}{2}}}\mathrm{d}x ],
\end{equation}
where $\phi$ denotes the spatial density of gNBs following a homogeneous Poisson point process \cite{rayleigh}.

\subsubsection{Bandwidth allocation}
To balance the tradeoff between transmission reliability and energy consumption, we first optimize $E^t$ with regard to bandwidth resource allocation, and then an adaptive transmit power control scheme is proposed to minimize $ \hat{Q}_k^t$ based on $E^t$. Thus, the subproblem of bandwidth allocation can be reformulated as
\begin{equation}\label{eqn:opt3}
 \min_{\bm{\delta}^t} \sum_{n=1}^N \beta_{n,k}^{t} P_k^t E^t, \,\,s.t. \,\, \eqref{eq:sub7}\,\, \mathrm{and} \, \,\eqref{eqn:trans_time}.
\end{equation}

%In Theorem \ref{theorem2}, the bandwidth allocation is optimized.

\begin{theorem}\label{theorem2}
Denote the optimal $E^t$ of $\eqref{eqn:opt3}$ as $E^{t*}= \frac{G_{n,k}^{\Delta \mathbf{w}}}{W \delta_{n,k}^t \log_2(1+ \mathbf{SINR} )}$, $\forall k$. The optimal bandwidth allocation solution to \eqref{eqn:opt3} satisfies
\begin{equation}\label{eqn:theorem2}
\sum_{k=1}^K \sum_{n=1}^N \beta_{n,k}^{t} \delta_{n,k}^{t*}=1.
\end{equation}
\end{theorem}
\begin{proof}
Denote by $\mathbf{S} = |\mathbf{h}_k^t|^2 P_k^t d_k^{-\alpha}$, and $\mathbf{I} = \sum_{U_{l} \in \tilde{\Xi}_{o,k}} |\mathbf{h}_l^t|^2 P_l^t d_l^{-\alpha}$, the inequality exists
\begin{equation}\label{eqn:proof_1}\small
\begin{aligned}
 & \frac{\mathrm{d}}{\mathrm{d}\delta}(W \delta \log_2(1+ \frac{\mathbf{S}}{\mathbf{I} + W N_0 \delta} )) \\
 = & \frac{W}{\mathrm{ln}2}[ \mathrm{ln}( 1 + \frac{\mathbf{S}}{\mathbf{I} + W N_0 \delta})
 - \frac{\mathbf{S} W N_0 \delta}{ (\mathbf{S}+ \mathbf{I} + W N_0 \delta) (\mathbf{I} + W N_0 \delta)} ] \overset{(a)} >0,
\end{aligned}
\end{equation}
where $(a)$ is based on $\mathrm{ln}(1+x) > \frac{x}{1+x}$, $\forall x>0$. For $0\leq \delta \leq 1$, it holds $\frac{\mathrm{d} E}{\mathrm{d}\delta} \leq 0$, and thus $E^{t*}$ monotonically decreases with $\delta_{n,k}^t$. The optimal solution to \eqref{eqn:opt3} is achieved if and only if all bandwidth is allocated \cite{binary}, which corroborates \eqref{eqn:theorem2}.
\end{proof}
%\begin{figure*}[ht]
%\begin{equation}\label{eqn:proof_theorem_2} \small \vspace*{-3mm}
% \frac{\mathrm{d}}{\mathrm{d}\delta} [W \delta \log_2(1+ \mathbf{SINR} )] = \frac{W}{\mathrm{ln}2} \mathrm{ln}( 1 + \frac{|\mathbf{h}_k^t|^2 P_k^t d_k^{-\alpha}}{\mathbf{I} + W N_0 \delta}) - \frac{W}{\mathrm{ln}2} \frac{|\mathbf{h}_k^t|^2 P_k^t d_k^{-\alpha} W N_0 \delta}{ (|\mathbf{h}_k^t|^2 P_k^t d_k^{-\alpha}+ \mathbf{I} + W N_0 \delta) (\mathbf{I} + W N_0 \delta)} \overset{(a)} >0,
%\end{equation}\normalsize  \vspace*{-4mm}
%\end{figure*}

Based on Theorem \ref{theorem2}, $\delta_{n,k}^{t*}$ of \eqref{eqn:opt3} can be derived in the form of Lambert-$W$ function $\mathbf{LW}(\cdot)$ as
\begin{equation}\label{eqn:optimal_bandwidth}
\delta_{n,k}^{t*} = \frac{|\mathbf{h}_k^t|^2 P_k^t d_k^{-\alpha} }{W N_0( e^{-\mathbf{LW}(\varphi)} - 1) } - \frac{\sum_{U_{l} \in \tilde{\Xi}_{o,k}} |\mathbf{h}_l^t|^2 P_l^t d_l^{-\alpha}}{W N_0},
\end{equation}
%However, due to the existence of interference, the tractable expression of $\delta_{n,k}^t$ cannot be directly obtained. Therefore, by employing adaptive power control and interference coordination techniques to mitigate the inter-cell interference, the optimal bandwidth allocation $\delta_{n,k}^{t*}$ of \eqref{eqn:opt3} can be approximated as
%the techniques such as Successive Interference Cancellation (SIC), the interference power can be approximated negligibly compared to the noise power, and we can derive
%\begin{equation}\label{eqn:optimal_bandwidth} \vspace*{-3mm}
%\delta_{n,k}^{t*} = \frac{G_{n,k}^{\Delta \mathbf{w}} \mathrm{ln}2}{ E^{t*} [\mathbf{W}( \varphi_{n,k}^t e^{\varphi_{n,k}^t} ) + \varphi_{n,k}^t ] },
%\end{equation}
where $\varphi$ can be approximated by the following implicit equation through numerical algorithms such as Newton's method
\begin{equation}\label{eqn:varphi}\small
\varphi = [ \frac{  G_{n,k}^{\Delta \mathbf{w}}N_0 \mathrm{ln}2 }{E^{t*}|\mathbf{h}_k^t|^2 P_k^t d_k^{-\alpha} } - \frac{\sum_{U_{l} \in \tilde{\Xi}_{o,k}} |\mathbf{h}_l^t|^2 P_l^t d_l^{-\alpha}}{|\mathbf{h}_k^t|^2 P_k^t d_k^{-\alpha}} \mathbf{LW}(\varphi) ]( \frac{\varphi}{\mathbf{LW}(\varphi)} - 1).
\end{equation}
%(***)When the interference is mitigated, the solution in \eqref{eqn:optimal_bandwidth} falls back to the result in \cite{binary}.
%(***)Then we propose a bandwidth allocation algorithm with adaptive transmission time length to jointly enhance the bandwidth usage while reducing the transmission overheads.

%Due to that the optimal $\delta_{n,k}^{t*}$ cannot be analytically solved in \eqref{eqn:theorem2},
As the analytical solution for $\delta_{n,k}^t$ cannot be obtained with Lambert-$W$ function, a two-tier binary search-based algorithm is designed to approach the optimal solution numerically in Algorithm \ref{alg:binary}.
%As $\delta_{n, k}^t$ monotonically decreases with $E^t$, we can adjust $E^t$ to approach $\delta_{n, k}^{t*}$.
On the first tier, the LoRA module-specific bandwidth $\delta_{n,k}^t$ is derived, and then aggregated based on model switching results $\beta_{n,k}^t$ on the second tier. If the overall bandwidth $\bm{\delta}_\mathrm{A}$ exceeds the threshold, the searching region for $E^t$ is halved by retaining the larger half, and vice versa until reaching the convergence requirement.
%in order to find a numerically efficient solution with the complexity on the order of $\mathcal{O}(K\mathrm{log}_2(T^\mathrm{input}_\mathrm{th}/\theta))$.

%In order to make full use of bandwidth resources and lower latency threshold for decrement in outage probability, if the overall bandwidth is redundant, the searching region for $E$ is reduced by half while maintaining the lower part, and the higher half region is retained if the overall bandwidth is insufficient. The searching process will terminate when the requirement $\theta$ is satisfied.

%It should be noted that not all the UEs have the same successful participation probabilities, in which the UEs with worse channel conditions will be allocated less bandwidth. with an outage probability as low as possible
\begin{algorithm}[!h]
\caption{A two-tier binary search algorithm for bandwidth allocation}
	\begin{algorithmic}[1]\label{alg:binary}
		%\STATE \textbf{Input}: $\beta_{n,k}^t$, and $E_\mathrm{input}$.
        \STATE \textbf{Initialization}: $E_\mathrm{up}^t=E_\mathrm{input}$, $E_\mathrm{down}^t = E_{\mathrm{min}}$, $E^t = \frac{1}{2}(E_\mathrm{up}^t + E_\mathrm{down}^t)$, precision coefficient $\omega$, and the maximum number of iterations $j_{\mathrm{max}}$ ;
        \STATE \textbf{Repeat:} For the $j$-th iteration, $1 \leq j \leq j_{\mathrm{max}}$ \\
        \begin{itemize}
            \item For the device $U_k$, $k = 1, \cdots, K$; \\
            $\textit{(Tier~1.~Solve~LoRA~module-specific~bandwidth)}$
            \begin{itemize}
            \item For the $n$-th LoRA module;
                \begin{itemize}
                    \item Derive $\varphi$ by solving \eqref{eqn:varphi} based on $P_k^t$ and $E^t$.\\
                    \item Derive $\delta_{n,k}^{t}$ by solving \eqref{eqn:optimal_bandwidth} based on $\varphi$.\\
                \end{itemize}
            \end{itemize}
            $\textit{(Tier~2.~Solve~UE-specific~bandwidth)}$
            \item Update $\bm{\delta}_\mathrm{A} =\sum_{k=1}^K \sum_{n =1}^N \beta_{n,k}^t \delta_{n,k}^t$.
            \begin{itemize}
            %\item If $1-\theta \leq \bm{\delta}_\mathrm{A} \leq 1$, record $\delta_{n,k}^t$ and $E^t$, $\forall n$, $\forall k$.
            \item If $\bm{\delta}_\mathrm{A} > 1$, $E_\mathrm{down}^t = E^t$, $E^t=\frac{1}{2}(E^t + E_\mathrm{up}^t)$.
            \item If $\bm{\delta}_\mathrm{A} <1-\omega$, $E_\mathrm{up}^t = E^t$, $E^t=\frac{1}{2}(E^t + E_\mathrm{down}^t)$.
            \end{itemize}
        \end{itemize}
        \STATE \textbf{Termination:} When $1-\omega \leq \bm{\delta}_\mathrm{A} \leq 1$ or $j > j_{\mathrm{max}}$. \\
        \STATE \textbf{Return:} $\delta_{n,k}^t$, $\forall n$, $\forall k$, and $E^t$.
	\end{algorithmic}
\end{algorithm}
\begin{algorithm}[htb]
\caption{A joint online optimization algorithm of \eqref{eqn:opt}}
	\begin{algorithmic}[1]\label{alg:alg2}
        \STATE \textbf{Initialization}: $\beta_{n,k}^{1} = 1$, $\delta^{1}_{n,k}$, $E^{1}$, $P^{1}_{k}$, $\forall n$, $\forall k$, $\bm{\lambda}^0 = \mathbf{0}$, maximum iteration $e_{\mathrm{max}}$, and convergence threshold $q_{\mathrm{th}}$.
        \STATE For the $t$-th round, $t=1,\cdots,T$, \\
        \quad ~~$\textit{(Step~1.~Wireless~resource~management)}$
        \begin{itemize}
            \item \textbf{Repeat:} For the $e$-th iteration, $1 \leq e \leq e_{\mathrm{max}}$;
            \begin{itemize}
            \item Initialize $\delta^{t,1}_{n,k} = \delta^{t}_{n,k}$, $E^{t,1} = E^{t}$, and $P^{t,1}_{k}=P^{t}_{k}$.
            % \item Calculate $\alpha^{e}_{n,k}$ based on $\delta^{t,e}_{n,k}$, $E^{t,e}$ and $P^{t,e}_{k}$.
            \item Update $\delta_{n,k}^{t,e+1}$ and $E^{t,e+1}$ by invoking Algorithm \ref{alg:binary} based on $\beta_{n,k}^{t}$ and $E^{t,e}$.
            \item Update $P_k^{t,e+1}$ by solving \eqref{eqn:transmit_power} based on $\beta_{n,k}^{t}$, $\delta_{n,k}^{t,e+1}$ and $E^{t,e+1}$.
            \item Calculate the value of $\hat{Q}_k^{t,e+1}$ in \eqref{eqn:main_obj}      %  based on the update results of $P_k^{t,e+1}$, $\beta_{n,k}^{t}$, $\delta_{n,k}^{t,e+1}$ and $E^{t,e+1}$.
            \end{itemize}
        \item \textbf{Termination:} When $| \hat{Q}_k^{t,e+1} - \hat{Q}_k^{t,e}| \leq q_{\mathrm{th}}$, $\forall k$, or $e > e_{\mathrm{max}}$.\\
        %\item \textbf{Return:} $\delta_{n,k}^{t+1}$, $E^{t+1}$, $P^{t+1}_{k}$, and obtain $Q_k^{t+1}$ in \eqref{eqn:main_obj}.\\
        $\textit{(Step~2.~Model~switching~optimization)}$
        \item Update $\beta_{n,k}^{t+1}$ and $\bm{\lambda}^{t+1}$ by solving \eqref{eqn:opt_4} based on $\delta_{n,k}^{t+1}$ and $P^{t+1}_{k}$, and derive the integer solutions.
        %$P_k^{t, e_{\mathrm{max}}}$, $\delta_{n,k}^{t, e_{\mathrm{max}}}$ and $E_{\mathrm{th}}^{t,e_{\mathrm{max}}}$.
            %\item Update $\mu_n^{t}$ by solving \eqref{eqn:opt_5} based on $\beta_{n,k}^{t}$, and implement local inference in \eqref{eqn:inference}.
            %\item Calculate $c_k^{t+1}$ in \eqref{eqn:suc_recovery} based on $\beta_{n,k}^{t}$ and $P_k^{t}$, and implement edge fine-tuning in \eqref{eqn:global_model}.
        \end{itemize}
        %\STATE \textbf{Output}: $\varphi_{i,t}$, $\forall U_i \in \mathcal{S}_{m,t}$, and $E$.
	\end{algorithmic}
\end{algorithm}

\subsubsection{Transmit power control}
The optimization subproblem for $P_k^t$ can be expressed as
\begin{equation}\label{eqn:transmit_power}
\min_{\mathbf{P}^t} \hat{Q}_k^t(P_{k}^t)\,\, \mathrm{in} \,\,\eqref{eqn:main_obj},\,\, s.t.~\eqref{eq:sub1}.
\end{equation}

The transmit power $P_k^t$ is first relaxed as a continuous real number, and thus \eqref{eqn:transmit_power} can be converted into a single-variate optimization problem. The first-order derivative of $\hat{Q}_k^t(P_k^t)$ with respect to $P_k^t$ is obtained as \eqref{eqn:derivate}.
\begin{figure*}[ht]
\begin{equation}\label{eqn:derivate}\small \vspace*{-3mm}
\nabla_P \hat{Q}_k^t(P_k^t) = -  \sum_{n=1}^N (1-\beta_{n,k}^t)( \frac{16\xi \zeta_2}{\epsilon} B_{n,k}^{t-1} +\frac{4\zeta_1}{\epsilon}) \frac{W \delta_{n,k}^t N_0 \theta d_i^2}{(P_k^t)^2 (2 \phi)^{\frac{\alpha}{2}-1}} e^{-\frac{W \delta_{n,k}^t N_0 \theta d_i^2}{P_k^t (2 \phi)^{\frac{\alpha}{2}-1}}} + \mu \sum_{n=1}^N \beta_{n,k}^t E^t
\end{equation}\normalsize \vspace*{-3mm}
\end{figure*}
Since directly solving $\nabla_P \hat{Q}_k^t(P_k^t)=0$ has no tractable closed-form solutions, the stationary point cannot be acquired. Therefore, the numerical algorithms such as interior point method are employed to improve computational efficiency.

%Recalling \eqref{eqn:theorem}, the weight of model integration should be sophisticatedly designed to mitigate the impact of distribution divergence \cite{noniid}, in which the model that obtains lower risk gap should be given higher weight and vice versa. Therefore, at each $t$, the subproblem can be expressed as

%As \eqref{eqn:opt_5} is a linear programming problem, we can obtain its optimal solution efficiently by using various optimization tools, such as the open source optimization package provided by CVXPY \cite{CVX}.
In Algorithm \ref{alg:alg2}, we design an online optimization algorithm of \eqref{eqn:opt} by jointly solving the aforementioned three subproblems. In the $t$-th round, the variables of transmit power and bandwidth allocation are updated iteratively until the objective $Q  _k^t$ converged, and then the model switching solutions are updated to minimize the cumulative cost. % towards the global optimal solution

%A xAPP is configured by the Near-RT RIC to manage the federated learning process within a cluster via standard interfaces.

%In this paper, the synchronous federated learning is modelled, in which the UEs with high latency will not be permitted to participate in federated learning.
%Due to the high processing capabilities at gNodeB, the computation latency for gradient aggregation and the transmission latency for downlink broadcasting are not modelled.

%In the wireless federated learning scenario, the transmitted signals from UEs will suffer from interference from the UEs in other cells, resulting in impaired channel quality.

%By applying Slivnyak’s theorem \cite{Stochastic2} to the stationary PPP of gNodeBs,

\addtolength{\topmargin}{0.05in}

%Since the synchronous federated learning may suffer from training inefficiency and large overheads due to the high computation and communication latency, the UEs should be scheduled by a specific policy in which a certain number of UEs are chosen to participate in federated learning.
%Moreover, the computation latency $\delta^{comp}_{i,t}$ of $U_i$ is characterized by shifted exponential distribution by following \cite{shifted}.

% straightforwardly
%\addtolength{\topmargin}{-0.02in}

%However, due to distribution divergence, the accuracy performance of federated learning model trained in a specific cell can not be guaranteed on new wireless settings.

%In order to derive accurate inference results for unknown data on new nodes in the O-RAN networks, the model ensemble [XX] method is employed which can effectively reduce the accuracy gap in model inference by aggregating multiple learning models based on the weights and improve model generality.

%\vspace*{-4mm}
\section{Simulation Results}
In this section, our proposed scheme is evaluated in a federated fine-tuning system to implement language sentiment analysis and  language inference tasks based on SST-2 and QNLI data set, respectively. We configure $N=4$ LoRA modules based on RoBERTa model \cite{FlexLora}. The system parameters are set as $f_k = 1.5$ GHz, $\alpha= 3.8$, $s_t = 1$, $\tau = 4$, $v_n = 0.1$, $W=1$ GHz, $\varrho_k=1 \times 10^{-27}$, $M_k=737.5$ cycles/bit, $N_0= -162$ dBm/Hz, $P_{\mathrm{max}} = 200$ mW, and $E_{\mathrm{min}} = 5$ ms. To characterize the features of heterogenous settings, the data of 100 UEs is drawn from a Dirichlet distribution \cite{HierFlexLora}.

%zero-shot
%Recall-Oriented Understudy for Gisting Evaluation Rogue-L  Score

% The models applied are CNN and ResNet, and the model parameters are generated based on Gaussian distribution.

% The scheme that transmits with normalized power $P_{\mathrm{max}}$ is employed as the upper bound. It can be observed that the energy consumption keeps rising as $\varepsilon_{\mathrm{th}}$ increases, especially in the high $\varepsilon_{\mathrm{th}}$ region.
\begin{figure*}[htb]
   \begin{minipage}{0.26\textwidth}
     \centering
     \includegraphics[width=1.0\linewidth]{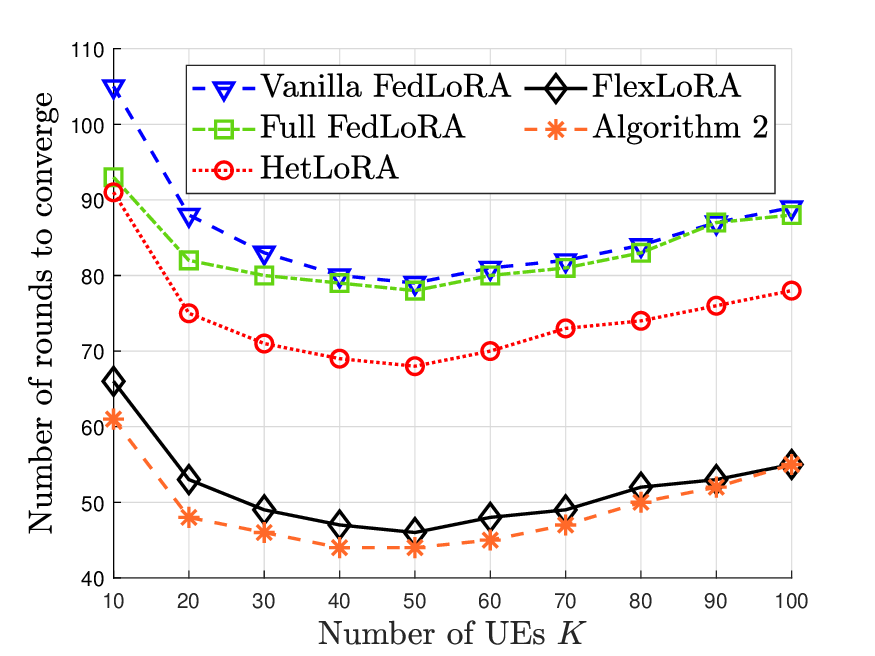}
     \vspace*{-6mm}\caption{Convergence performance\\ of Algorithm \ref{alg:alg2} (SST-2 data set).}\label{fig:figure1}
   \end{minipage}\hspace*{-1.0em}
   \begin{minipage}{0.26\textwidth}
     \centering
     \includegraphics[width=1.0\linewidth]{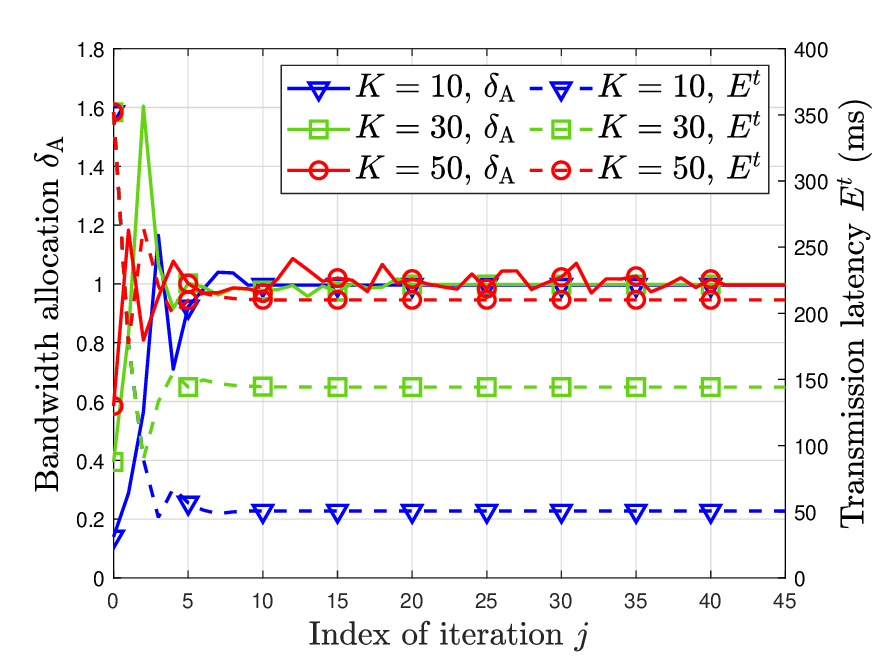}
     \vspace*{-6mm}\caption{Stability performance of \\Algorithm \ref{alg:binary} (SST-2 data set).}\label{fig:figure2}
   \end{minipage}
   \begin{minipage}{0.26\textwidth}
     \centering
     \includegraphics[width=1.0\linewidth]{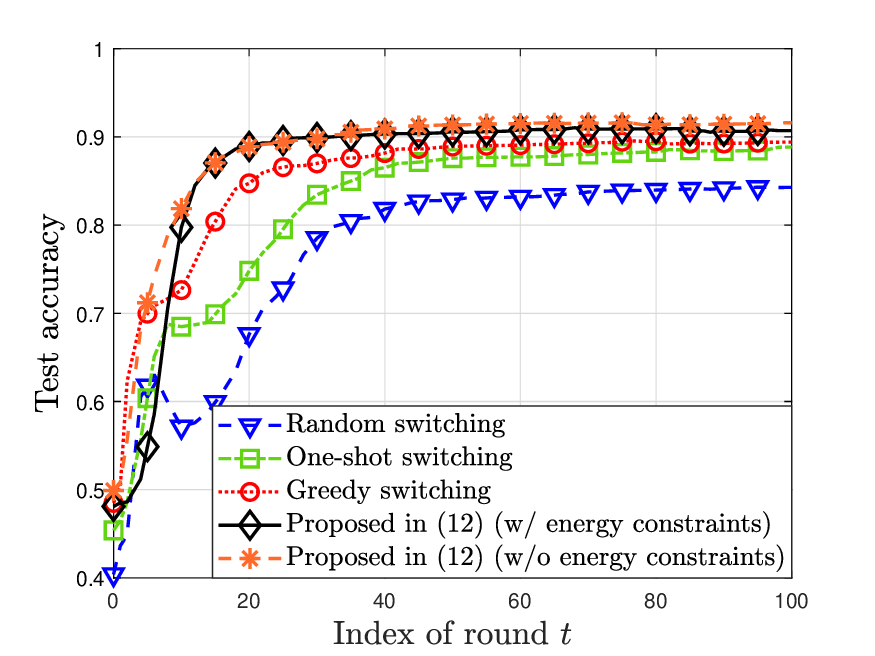}
     \vspace*{-6mm}\caption{Test accuracy performance\\ with model switching schemes \\(QNLI data set).}\label{fig:figure3}
   \end{minipage}\hspace*{-1.0em}
   \begin{minipage}{0.26\textwidth}
     \centering
     \includegraphics[width=1.0\linewidth]{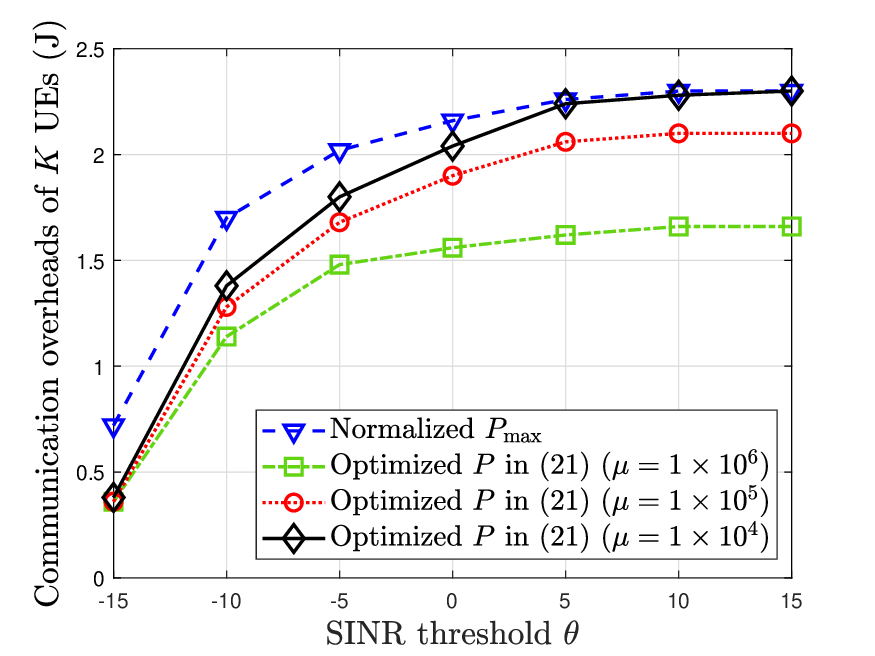}
     \vspace*{-6mm}\caption{Energy efficiency perform-\\ance of Algorithm \ref{alg:alg2} with coefficient\\ $\mu$ (QNLI data set).}\label{fig:figure4}
   \end{minipage}\hspace*{-1.0em}  \vspace*{-4mm}
\end{figure*}

In Fig. \ref{fig:figure1}, the convergence performance versus the number of UEs $K$ is evaluated with four benchmarks, including vanilla FedLoRA \cite{few-shot}, full-parameter FedLoRA, HetLoRA \cite{hetlora} and FlexLoRA \cite{FlexLora}. The number of rounds to converge first declines in the lower $K$ region and then increases, which showcases the tradeoff between scheduling more UEs and allocating more bandwidth resources to enhance transmission reliability. By employing multiple LoRA modules for adapting to the non-IID data characteristics, Algorithm \ref{alg:alg2} can speed up convergence by $42.7\%$ on average, compared with FedLoRA which fine-tunes a common LoRA module. In addition, our proposed scheme outperforms HetLoRA and FlexLoRA by providing each UE a tailored LoRA module for federated fine-tuning with device heterogeneity, and the LoRA module configuration is dynamic via model switching to overcome the impact of time-varying wireless circumstances. Thus, the performance gains can be further enlarged by $34.1\%$ and $5.0\%$, respectively.

% By integrating the personalized knowledge from downstream tasks in a federated manner, our proposed scheme can well improve the inference performance and accelerate convergence rate.
% Furthermore, the gains of federated fine-tuning can be further enlarged with a higher rank $r=8$, which showcases the feasibility of designing edge device-specific LoRA modules based on the tradeoff between generalization capability and resource overheads to realize optimal fine-tuning performance.

In Fig. \ref{fig:figure2}, the results show that Algorithm \ref{alg:binary} can always attain a stable bandwidth allocation solution. As $K$ rises, the iteration to converge also increases, which manifests the time complexity of Algorithm \ref{alg:binary} follows the order of $\mathcal{O}(K\mathrm{log}_2(E_\mathrm{input}/\omega))$. Moreover, the transmission latency is adaptive to $K$, which corroborates the scalability of Algorithm \ref{alg:binary} in diverse wireless circumstances to jointly enhance bandwidth usage while reducing communication overheads.

Fig. \ref{fig:figure3} demonstrates the test accuracy on QNLI data set with model switching schemes, in which the greedy scheme optimizes short-term slice of \eqref{eq:sub9a} in each round, and the one-shot scheme implements model switching once and follows vanilla FedLoRA thereafter. By minimizing the long-term generalization risk gap with device heterogeneity, the convergence performance of our proposed scheme is guaranteed, and the test accuracy improves by $1.9\%$ and $1.3\%$ compared with one-shot and greedy schemes, respectively. Moreover, the tradeoff between generalization capability and energy efficiency is achieved via energy consumption-adaptive switching (w/ energy constraints) with a $0.9\%$ reduction in test accuracy.

%As our proposed method in \eqref{eqn:opt_4} minimizes the long-term power consumption with model complexity
%As $t$ rises, the performance gap with the greedy method can be further enlarged as the long-term constraint in \eqref{eq:sub4} can avoid high-consumption switching operations per round.
%Furthermore, as a larger $v_n$ motivates the long-term participation of more models and leads to frequent model switching, the inference accuracy performance is harmed.

% Moreover, the impact of minimum participation rate constraint $v_n$ in \eqref{eq:sub4} is evaluated. It can be observed that the scheme with $v_n=0.4$ outperforms $v_n=0.1$, and the gap is lowered when $N$ increases.

In Fig. \ref{fig:figure4}, the communication overheads of Algorithm \ref{alg:alg2} are evaluated with respect to the transmit power control scheme and energy efficiency coefficient $\mu$. Unlike the maximum transmit power scheme with $P_k^t =P_{\mathrm{max}}$, the transmit power in \eqref{eqn:transmit_power} can be optimized adaptively based on wireless channel condition $\theta$. Therefore, the power consumption can be saved by $16.8\%$ when $\theta=-5$ dB and $\mu = 1 \times 10^5$, and the performance gap is mitigated when the channel quality deteriorates. As the impact of energy consumption increases in the larger $\mu$ region, the communication overheads have been reduced by $11.9\%$ when $\mu = 1 \times 10^6$, and the energy efficiency requirement can be guaranteed.

\section{Conclusion}
In this paper, a federated fine-tuning paradigm via online optimization over wireless networks is proposed, in which the UEs employ a model-switching scheme and subscribe LoRA modules dynamically to implement federated fine-tuning. The theoretical analysis has proven an tractable upper bound on the generalization performance, and an online algorithm to mitigate the impact of device heterogeneity is designed. Finally, the simulation results verify the effectiveness of the proposed scheme in enhancing learning performance.
%, and a

\appendices
\section{Proof of Theorem 1}
Denote the global risk function as $F(\mathbf{w}) = \sum_{k=1}^K \rho_k^{t} F_k(\mathbf{w})$. Recalling the definition in \eqref{eqn:inference}, $\Phi_t$ is first rewritten as
\begin{equation}\label{eqn:proof_1}\small
\begin{aligned}
\Phi_t & = \frac{1}{N} \sum_{k=1}^K \rho_k^t \sum_{n=1}^N  \mathbb{E} [F_k(\Delta \mathbf{w}_{n}^{t+1}+\mathbf{w}_0)-F_k(\mathbf{w}^{*})] \\
& + \frac{1}{N} \sum_{k=1}^K \rho_k^t \sum_{n=1}^N \mathbb{E} [F_k(\mathbf{w}^{*})- F_k(\mathbf{w}_{k,t}^{*})],
\end{aligned}
\end{equation}
where $\mathbf{w}^{*} = \arg\min_{\mathbf{w}} F(\mathbf{w})$. The second-order Taylor expansion of $F_k(\Delta \mathbf{w}_{n}^{t+1}+\mathbf{w}_n)$ can be derived as
\begin{equation}\label{eqn:proof_1_1}\small
 F_k(\Delta \mathbf{w}_{n}^{t+1}+\mathbf{w}_0)
%= F_k(\Delta \mathbf{w}_{n}^{t+1}) + \mathbf{w}_n^\top \nabla F_k(\Delta \mathbf{w}_{n}^{+1t}) + \frac{1}{2}\mathbf{w}_n^\top \nabla^2 F_k(\Delta \mathbf{w}_{n}^{t+1}) \mathbf{w}_n
\leq F_k(\Delta \mathbf{w}_{n}^{t+1})  + \mathbf{w}_n^\top \nabla F_k(\Delta \mathbf{w}_{n}^{t+1}) + \frac{\epsilon}{2} \| \mathbf{w}_n \|^2,
\end{equation}
where the inequality follows Assumption \ref{assumpt2}.
% & = \underbrace{\sum_{k=1}^K \rho_k^t \sum_{n=1}^N \mu_n \beta_{n,k}^t [F_k(\mathbf{w}_{n}^t)-F_k(\mathbf{w}^{*})]}_{\Phi_1} + \sum_{k=1}^K \rho_k^t \sum_{n=1}^N \mu_n \beta_{n,k}^t  \\
% & [\underbrace{F_k(\mathbf{w}^{*})- F_k(\mathbf{\tilde{w}}_k^{*})}_{\Phi_2}]  + \underbrace{\sum_{k=1}^K \rho_k^t (1 - \sum_{n=1}^N \mu_{n,k}^t \beta_{n,k}^t) [ F_k(\mathbf{\tilde{w}}_{k}^{t}) - F_k(\mathbf{\tilde{w}}_k^{*})]}_{\Phi_3}.
%\addtolength{\topmargin}{0.07in}
Denote by $\mathbf{e}_n^t = \nabla F(\Delta \mathbf{w}_{n}^{t}) - \frac{\sum_{k=1}^K D_k^t \beta_{n,k}^t \gamma_{n,k}^t \nabla  F_{k}(\Delta \mathbf{w}_{n}^t)}{\sum_{k=1}^K D_k^t \beta_{n,k}^t \gamma_{n,k}^t}$, which satisfies $\Delta\mathbf{w}_{n}^{t+1} = \Delta \mathbf{w}_{n}^{t} - \eta (\nabla F(\Delta \mathbf{w}_{n}^{t}) - \mathbf{e}_n^t)$ in \eqref{eqn:global_model}.

To study the relationship between $F_k(\Delta \mathbf{w}_{n}^{t+1})$ and $F_k(\Delta \mathbf{w}_{n}^{t})$, $F_k(\Delta \mathbf{w}_{n}^{t+1})$ can be second-order expanded as
\begin{equation}\label{eqn:proof_2}\small
\begin{aligned}
%& F_k(\Delta \mathbf{w}_{n}^{t+1})
%\leq F_k(\Delta \mathbf{w}_{n}^{t}) \\
%& + (\Delta \mathbf{w}_{n}^{t+1} - \Delta \mathbf{w}_{n}^{t})^\top \nabla F_k(\Delta \mathbf{w}_{n}^{t}) + \frac{\epsilon}{2} \| \Delta \mathbf{w}_{n}^{t+1} - \Delta \mathbf{w}_{n}^{t} \|^2 \\
&F_k(\Delta \mathbf{w}_{n}^{t+1}) \leq F_k(\Delta \mathbf{w}_{n}^{t}) - \eta (\nabla F(\Delta \mathbf{w}_{n}^{t}) - \mathbf{e}_n^t)^\top \nabla F_k(\Delta \mathbf{w}_{n}^{t}) \\
& + \frac{\eta^2 \epsilon}{2} \| \nabla F(\Delta \mathbf{w}_{n}^{t})\|^2 - \eta^2 \epsilon (\mathbf{e}_n^t)^\top \nabla F(\Delta \mathbf{w}_{n}^{t}) + \frac{\eta^2 \epsilon}{2} \| \mathbf{e}_n^t \|^2.
\end{aligned}
\end{equation}
% which captures the influence of transmission unreliability and carbon footprint

\addtolength{\topmargin}{0.05in}
By taking the summation of \eqref{eqn:proof_2} over $K$ UEs weighted by $\rho_k^t$, and subtracting $\sum_{k=1}^K \rho_k^t F_k(\mathbf{w}^{*})$ in both sides with step size $\eta = \frac{1}{\epsilon}$, it can be derived that
%\begin{equation}\label{eqn:proof_5}\small
%\begin{aligned}
%& \sum_{k=1}^K \rho_k^t F_k(\Delta \mathbf{w}_{n}^{t+1}) - \sum_{k=1}^K \rho_k^t F_k(\mathbf{w}^{*}) \leq \sum_{k=1}^K \rho_k^t F_k(\Delta \mathbf{w}_{n}^{t}) - \eta \nabla F(\Delta \mathbf{w}_{n}^{t})^\top \\
%& \sum_{k=1}^K \rho_k^t \nabla F_k(\Delta \mathbf{w}_{n}^{t}) + \sum_{k=1}^K \rho_k^t (\mathbf{e}_n^t)^\top ( \eta \nabla F_k(\Delta \mathbf{w}_{n}^{t}) - \eta^2 L \nabla F(\Delta \mathbf{w}_{n}^{t})) \\
%&  + \frac{\eta^2L}{2} \sum_{k=1}^K \rho_k^t \| \nabla F(\Delta \mathbf{w}_{n}^{t})\|^2 + \frac{\eta^2L}{2} \sum_{k=1}^K \rho_k^t \| \mathbf{e}_n^t \|^2 - \sum_{k=1}^K \rho_k^t F_k(\mathbf{w}^{*}).
%\end{aligned}
%\end{equation}
\begin{equation}\label{eqn:proof_6}\small
\begin{aligned}
& \sum_{k=1}^K \rho_k^t [(F_k(\Delta \mathbf{w}_{n}^{t+1}) - F_k(\mathbf{w}^{*}))- (F_k(\Delta \mathbf{w}_{n}^{t}) - F_k(\mathbf{w}^{*}))]  \\
& \leq - \frac{1}{\epsilon} \nabla F(\Delta \mathbf{w}_{n}^{t})^\top \sum_{k=1}^K \rho_k^t \nabla F_k(\Delta \mathbf{w}_{n}^{t}) + \frac{1}{2 \epsilon} \| \nabla F (\Delta \mathbf{w}_{n}^{t})\|^2 \\
& + \frac{1}{\epsilon} \sum_{k=1}^K \rho_k^t (\mathbf{e}_n^t)^\top ( \nabla F_k(\Delta \mathbf{w}_{n}^{t}) - \nabla F (\Delta \mathbf{w}_{n}^{t})) + \frac{\| \mathbf{e}_n^t \|^2}{2 \epsilon},
\end{aligned}
\end{equation}
where $\| \mathbf{e}_n^t \|^2$ can be bounded based on Assumption \ref{assumpt3} as \cite{assum}
\begin{equation}\label{eqn:proof_8}\small
\mathbb{E}\| \mathbf{e}_n^t\|^2 \leq 4 \sum_{k=1}^K \rho_k^t (1-\lambda_{n,k})(1- \beta_{n,k}^t) (\zeta_{1} + \zeta_{2}\| \nabla F(\Delta \mathbf{w}_n^t)\|^2).
\end{equation}

The expectation of $(\mathbf{e}_n^t)^\top ( \nabla F_k(\Delta \mathbf{w}_{n}^{t}) - \nabla F(\Delta \mathbf{w}_{n}^{t}))$ can be extended by invoking $ \mathbb{E}(\| \mathbf{e}_n^t\|^2)$ in \eqref{eqn:proof_8} as
\begin{equation}\label{eqn:proof_9}\small
\begin{aligned}
& \mathbb{E}[\sum_{k=1}^K \rho_k^t (\mathbf{e}_n^t)^\top ( \nabla F_k(\Delta \mathbf{w}_{n}^{t}) - \nabla \bar{F}(\Delta \mathbf{w}_{n}^{t}))] \\
%& \leq \frac{1}{2} \sum_{k=1}^K \rho_k^t \{\mathbb{E}[\| \mathbf{e}_n^t \| ^2 ] + \mathbb{E}[\| \nabla F_k(\Delta \mathbf{w}_{n}^{t}) - \nabla \bar{F}(\Delta \mathbf{w}_{n}^{t}) \|^2] \}\\
%& = \frac{1}{2} \sum_{k=1}^K \rho_k^t \{\mathbb{E}[\| \mathbf{e}_n^t \| ^2 ] + \mathbb{E}[\| \frac{\sum_{l=1}^K D_l^{t}(\nabla F_k(\Delta \mathbf{w}_{n}^{t}) - \nabla F_l(\Delta \mathbf{w}_{n}^{t}))}{\sum_{l=1}^K D_l^{t}} \|^2] \} \\
& \overset{(a)} \leq \frac{1 }{2}\sum_{k=1}^K \rho_k^t  \{ \mathbb{E}\| \mathbf{e}_n^t\|^2 \\
& + \mathbb{E}[\frac{\sum_{j=1}^K (\mathcal{D}_{j, \mathrm{te}}^{t})^2 \sum_{j=1}^K \|\nabla F_k(\Delta \mathbf{w}_{n}^{t}) - \nabla F_j(\Delta \mathbf{w}_{n}^{t})\|^2}{(\sum_{j=1}^K \mathcal{D}_{j, \mathrm{te}}^{t})^2}] \},
\end{aligned}
\end{equation}
% $(a)$  $\mathbf{a}^\top \mathbf{b} \leq \frac{1}{2} (\| \mathbf{a}\|^2 + \| \mathbf{b}\|^2)$, and
where $(a)$ stems from Cauchy-Schwarz inequality. Based on \eqref{eqn:global_model}, we expand quadratic terms and combine like terms as
%$\nabla F_k(\Delta \mathbf{w}_{n}^{t}) = \frac{1}{\mathcal{D}_{k, \mathrm{te}}^{t}} \sum_{d=1}^{\mathcal{D}_{k, \mathrm{te}}^{t}} \nabla l(\Delta \mathbf{w}_{n}^{t};\mathbf{x}^d,\mathbf{y}^d)$, we expand quadratic terms and combine like terms as
\begin{equation}\label{eqn:proof_11}\small
\begin{aligned}
& \|\nabla F_k(\Delta \mathbf{w}_{n}^{t}) - \nabla F_j(\Delta \mathbf{w}_{n}^{t})\|^2 \\
% & \leq \frac{(D_{j, \mathrm{te}}^{t})^2 [\sum_{d=1}^{D_{k, \mathrm{te}}^t} \nabla l(\mathbf{w}_{n}^{t};\mathbf{x}^d)]^2 + (D_{k, \mathrm{te}}^{t})^2 [\sum_{d=1}^{D_{j, \mathrm{te}}^t} \nabla l(\mathbf{w}_{n}^{t};\mathbf{x}^d)]^2}{(D_{k, \mathrm{te}}^{t})^2 (D_{j, \mathrm{te}}^{t})^2} \\
& \leq \frac{ [\sum_{d=1}^{D_{k, \mathrm{te}}^t} \nabla l(\mathbf{w}_{n}^{t};\mathbf{x}^d)]^2}{(D_{k, \mathrm{te}}^{t})^2 } + \frac{ [\sum_{d=1}^{D_{j, \mathrm{te}}^t} \nabla l(\mathbf{w}_{n}^{t};\mathbf{x}^d)]^2}{(D_{j, \mathrm{te}}^{t})^2 } \\
& + \frac{2 D_{j, \mathrm{te}}^{t} D_{k, \mathrm{te}}^{t}\sum_{d=1}^{D_{k, \mathrm{te}}^t} \nabla l(\mathbf{w}_{n}^{t};\mathbf{x}^d) \sum_{d=1}^{D_{j, \mathrm{te}}^t} \nabla l(\mathbf{w}_{n}^{t};\mathbf{x}^d)}{(D_{k, \mathrm{te}}^{t})^2 (D_{j, \mathrm{te}}^{t})^2}\\
& \overset{(b)} \leq 4 (\zeta_{1} + \zeta_{2}\| \nabla F(\Delta \mathbf{w}_n^t)\|^2),
%& \leq \frac{(D_{l, \mathrm{te}}^{t})^2 (D_{k, \mathrm{te}}^t)^2 \|\nabla l(\Delta \mathbf{w}_{n}^{t};\mathbf{x},\mathbf{y}) \|^2 + (D_k^{t})^2  (D_l^{t})^2 \|\nabla l(\Delta \mathbf{w}_{n}^{t};\mathbf{x},\mathbf{y}) \|^2}{(D_k^{t})^2 (D_l^{t})^2} \\
%& + \frac{2 (D_l^{t})^2 (D_k^t)^2  \|\nabla l(\Delta \mathbf{w}_{n}^{t};\mathbf{x},\mathbf{y}) \|^2 }{(D_k^{t})^2 (D_l^{t})^2} \overset{(a)} \leq 4 (\zeta_{1} + \zeta_{2}\| \nabla F(\Delta \mathbf{w}_n^t)\|^2)
\end{aligned}
\end{equation}
where $(b)$ follows Assumption \ref{assumpt3}. As $\mathcal{D}_{j, \mathrm{te}}^{t} \geq 1$, $\forall j$, we have $\sum_{j=1}^K (\mathcal{D}_{j, \mathrm{te}}^{t})^2 \leq (\sum_{j=1}^K \mathcal{D}_{j, \mathrm{te}}^{t})^2$.
%and the item in \eqref{eqn:proof_9} satisfies that
%\begin{equation}\label{eqn:proof_11_1}\small
%\| \frac{\sum_{l=1}^K D_l^{t}(\nabla F_k(\Delta \mathbf{w}_{n}^{t}) - \nabla F_l(\Delta \mathbf{w}_{n}^{t}))}{\sum_{l=1}^K D_l^{t}} \|^2 \leq 4K(\zeta_{1} + \zeta_{2}\| \nabla F(\Delta \mathbf{w}_n^t)\|^2).
%\end{equation}
By invoking Assumption \ref{assumpt2}, the Polyak-Łojasiewicz inequality holds for $F(\mathbf{w})$ as \cite{assum}
\begin{equation}\label{eqn:proof_7}\small
\| \nabla F(\Delta \mathbf{w}_{n}^{t})\|^2 \geq 2 \xi (F(\Delta \mathbf{w}_{n}^{t}) - F(\mathbf{w}^{*})).
\end{equation}
\addtolength{\topmargin}{-0.02in}
Recalling the second term in \eqref{eqn:proof_1}, as $F_k(\mathbf{w}) $ is $\epsilon $-Lipschitz continuous in Assumption \ref{assumpt1}, it holds that
\begin{equation}\label{eqn:proof_12}\small
\mathbb{E}[F_k(\mathbf{w}^{*}) - F_k(\mathbf{w}_{k,t}^{*})] \leq \epsilon \sqrt{ \mathbb{E}\| \mathbf{w}^{*} - \mathbf{w}_{k,t}^{*} \|^2 }\overset{(c)} \leq \epsilon \sqrt{\frac{4 \zeta_1}{\xi^2 D_{k, \mathrm{te}}^{t}}},
\end{equation}
where $(c)$ follows Lemma 1 in \cite{lemma}. %For brevity of analysis, we set $\zeta_2 = 0$ in Assumption \ref{assumpt3} with a large $\zeta_1$.
By substituting \eqref{eqn:proof_8}, \eqref{eqn:proof_11}, \eqref{eqn:proof_7} and \eqref{eqn:proof_12} into \eqref{eqn:proof_6}, and applying it recursively for $t$ rounds, \eqref{eqn:theorem} can be derived. The proof has finished.\vspace*{-2mm}

%\section{Proof of Theorem 2}

\end{document}